\documentclass[sigconf,natbib=false,fleqn]{acmart}
\pdfoutput=1

\usepackage[utf8]{inputenc}

\usepackage{mathtools}

\usepackage[english]{babel}
\usepackage{csquotes} 
\usepackage{microtype}
\usepackage{xcolor}
\usepackage{booktabs}

\definecolor{darkblue}{rgb}{0,0,0.5}
\definecolor{cerule}{RGB}{53,122,183}
\definecolor{cardinal}{RGB}{184,32,16}
\hypersetup{
    colorlinks,
    linkcolor={black},
    citecolor={cerule},
    urlcolor={cardinal}
}

\allowdisplaybreaks[2]

\usepackage[
  citestyle=numeric-comp,
  bibstyle=numeric,
  backend=biber,
  isbn=false,
  maxcitenames=2,
  maxbibnames=6,
  giveninits=true,
  sortcites=true,
  doi=true,
  url=false
]{biblatex}
\AtEveryBibitem{\clearfield{month}}
\AtEveryBibitem{\clearfield{day}}
\DeclareNameAlias{sortname}{first-last}

\usepackage{float}
\usepackage[noend]{algpseudocode}
\algrenewcommand\textproc{}     

\usepackage[style=base]{caption}      

\floatstyle{ruled}
\newfloat{algo}{tp}{lop}
\floatname{algo}{Algorithm}
\usepackage{xspace}
\def\Return{\textbf{return}\xspace}

\usepackage[shortlabels]{enumitem} 
\setlist{nosep}

\def\eqdef{=}

\def\mop{\operatorname}
\def\st{\ \middle|\ }
\def\abs#1{\left|#1\right|}

\def\epsilon{\varepsilon}
\def\ud{\mathrm{d}}

\def\crit{\operatorname{crit}}

\def\Q{\mathbb{Q}}
\def\R{\mathbb{R}}

\def\whitney#1{\mathcal{W}{#1}}

\usepackage{lipsum}

\title{Computing the Dimension of Real Algebraic Sets}

\author{Pierre Lairez}
\orcid{}
\affiliation{
\institution{Inria}
\country{France}
}
\email{pierre.lairez@inria.fr}

\author{Mohab {Safey El Din}}
\affiliation{%
  \institution{Sorbonne Universit\'e, \textsc{CNRS},
      \textsc{LIP6},
    \'Equipe \textsc{PolSys}}
  \country{France}
}
\email{mohab.safey@lip6.fr}

\setcopyright{acmlicensed}

\copyrightyear{2021}
\acmYear{2021}
\setcopyright{acmlicensed}\acmConference[ISSAC '21]{Proceedings of the 2021 International Symposium on Symbolic and Algebraic Computation}{July 18--23, 2021}{Virtual Event, Russian Federation}
\acmBooktitle{Proceedings of the 2021 International Symposium on Symbolic and Algebraic Computation (ISSAC '21), July 18--23, 2021, Virtual Event, Russian Federation}
\acmPrice{15.00}
\acmDOI{10.1145/3452143.3465551}
\acmISBN{978-1-4503-8382-0/21/07}



\begin{abstract}
  Let $V$ be the set of real common solutions to $F = (f_1, \ldots, f_s)$ in
  $\R[x_1, \ldots, x_n]$ and $D$ be the maximum total degree of the $f_i$'s. 
  We design an algorithm which on input $F$ computes the dimension of $V$.
  Letting $L$ be the evaluation complexity of $F$ and $s=1$, it runs
  using $O^\sim \big (L D^{n(d+3)+1}\big )$ arithmetic operations in $\Q$
  and at most $D^{n(d+1)}$ isolations of real roots of polynomials of degree at
  most $D^n$.

  Our algorithm depends on the \emph{real} geometry of $V$; its practical
  behavior is more governed by the number of topology changes in the fibers of
  some well-chosen maps. Hence, the above worst-case bounds are rarely reached
  in practice, the factor $D^{nd}$ being in general much lower on practical
  examples. We report on an implementation showing its ability to solve problems
  which were out of reach of the state-of-the-art implementations.
\end{abstract}

\begin{document}
\fancyhead{}

\maketitle

\section{Introduction}\label{sec:intro}

Polynomial system solving over the reals finds numerous applications in
engineering sciences. In mechanism design, mobility properties often translate
to identifying the dimension of a real algebraic set, that is the real solution
set to polynomial equations with real coefficients. This paper focuses on the
design of an algorithm for computing dimension whose practical performance
improves upon previously known algorithms.

\subsubsection*{Prior results}
The cylindrical algebraic decomposition algorithm \cite{Collins_1975} computes a
partition of the semi-algebraic set defined by an input semi-algebraic formula
into finitely many cells which are homeomorphic to $]0,1[^i$ for some integer
$i$. This integer $i$ is the dimension of the corresponding cells. The dimension
of the whole semi-algebraic set is the maximum dimension of its cells. The
computational complexity of this approach is doubly exponential with respect to
the dimension of the ambient space, denoted~$n$.

Another approach \cite{Koiran_1997,Koiran_1999,BasuPollackRoy_2006} makes the
most of the following characterization of the dimension: the dimension of a
semi-algebraic set~$S \subseteq \mathbb{R}^n$ is the maximum integer $d$ such
that the projection of~$S$ on some coordinate subspace of dimension $d$ has
non-empty interior. Quantifier elimination over the reals computes a
semi-algebraic description of the projection of a semi-algebraic set. Hence,
using quantifier elimination algorithms and their complexity analysis, one can
derive the following complexity result. When the input is a single polynomial of
total degree $D$ in $\R[x_1, \ldots, x_n]$, computing the dimension $d$ of its
real vanishing set is performed with $D^{O(d(n-d))}$ arithmetic operations in
$\R$. Note that several polynomial constraints~$f_1,\dotsc,f_s$ can be
substituted by the single constraint~$\sum_i f_i^2$, although in practice we try
not to. A probabilistic variant refines this complexity bound to $O^\sim\big
(n^{16}(D+1)^{3d(n-d)+5n+5}\big )$ arithmetic operations and enables practical
implementations \cite{BannwarthSafeyElDin_2015}.

The dimension of a semi-algebraic set~$S$ is also the Krull dimension of its
real radical \cite{CoTr98,SafeyElDinYangZhi_2021}. Such methods compute a prime
decomposition of the radical ideal of the polynomials vanishing on $S$ and
decide whether associated irreducible components contain regular real points,
yielding a smooth point in $S$. If not, one needs to investigate the singular
loci of the algebraic sets defined by these components. Following this idea of
computing smooth points in~$S$, a dedicated numerical routine procedure for
computing smooth points in semi-algebraic sets based on numerical homotopy and
deflation algorithms is designed in \cite{HarrisHauensteinSzanto_2020}.
Complexity bounds are worse than the state-of-the-art bounds but the numerical
flavour of this algorithm makes it practically fast on well-conditioned
examples.

Lastly, the related problem of computing the local dimension at a point has also been considered
\cite{Vorobjov_1999}.

\subsubsection*{Problem statement} Current state-of-the-art algorithms are not satisfactory from
a practical view: the best algorithms, from a complexity viewpoint, still rely
on quantifier elimination which involves the computation of algebraic or
semi-algebraic formulas whose sizes are $D^{O(d(n-d))}$. This proves to be a
bottleneck. New numerical procedures have emerged but implementations are
vulnerable to accuracy issues and hit the problem of deciding equality to zero
from approximations.

The goal of this paper is to design a new algorithm for computing the dimension
of a \emph{real algebraic set}, practically efficient, which computes (and
stores) algebraic data of size bounded by $(sD)^{O(n)}$.

\subsubsection*{Main results}
Our algorithm takes as input polynomials $f_1, \ldots, f_s$ in $\R[x_1, \ldots,
x_n]$. Let~$D = \max _i \deg(f_i)$ and $V \subseteq \R^n$ be the real algebraic
set defined by $f_1=\cdots=f_s=0$. For a sufficiently generic function~$h \in
\mathbb{R}[x_1,\dotsc,x_n]$, we show that we can compute finitely many
points~$t_0,\dotsc,t_N \in \mathbb{Q}$ such that
\begin{equation}\label{eq:main-idea}
\dim V = \max_{0\leq i \leq N} \dim \big(V\cap h^{-1}(t_i)\big) + 1,
\end{equation}
unless~$V$ is empty.
When~$V$ is compact, we choose~$h$ to be a linear form. This leads to an induction on the number of variables.
In the general case, we choose~$h$ to be the squared distance to a generic point. This reduces the general case to the compact case.

When~$s=1$,
the points~$t_i$ are chosen so that~$t_{i-1} < v_i < t_{i+1}$
where the points~$v_i \in \mathbb{R}$ are limits of critical values of the restriction of~$h$ to the set~$\left\{ f = \epsilon \right\}$ as~$\epsilon \to 0$.
In particular, we prove that $N \leq D^n$.
The general case is more involved but analogue.

At each recursive call, the cost of the algorithm
lies in:
\begin{enumerate}[(i)]
  \item\label{item:1} deciding~$V \neq \varnothing$, in order to apply~\eqref{eq:main-idea};
  \item\label{item:2} computing a univariate polynomial whose zero set contains~$\{ v_1,\dotsc,v_N \}$;
  \item isolating the zeros of this polynomial to find appropriate points~$t_i$.
\end{enumerate}
The steps~\ref{item:1} and~\ref{item:2} are performed in~$(sD)^{O(n)}$
arithmetic operations. The worst case complexity of the last step depends on the
height of the coefficients of the input equations; we will not enter into such
considerations. The total number of recursive calls is~$(sD)^{O(nd)}$.

Our results do not improve on the state-of-the-art asymptotic complexity bound
which is exponential in $d(n-d)$ while we only reach $nd$ (and not counting the cost of
isolating the real roots of some univariate polynomials).
However, the behavior of the algorithm depends more on the
\emph{real geometry} of $V$ than what is observed for algorithms based on
quantifier elimination. Nonetheless, the excellent practical behavior is easily explained.

For one part, the practical behavior of the algorithm is governed by the
actual degree of the algebraic varieties arising in the computation of the
limits of critical values. It is typically lower than the worst-case
bounds. For another part, the practical behavior is governed by the number
of recursive calls, which is~$(sD)^{O(nd)}$ where~$d$ is the dimension of~$V$.
With previous notations, if~$V$ is not empty, there are~$N+1$ direct recursive calls
(and each one may induce other indirect recursive calls naturally). While the polynomial
computed in step~\ref{item:2} has degree at most~$(6sD)^{n}$, the
actual degree is often lower than this bound. And the number~$N$ of \emph{real} roots is even lower
(see \S \ref{sec:experiments}). In our examples, the total number of recursive
calls ends up being dramatically lower than the worst-case bound. This accounts
for the good practical performance of our algorithm.

We report in detail on practical experiments which illustrate the interest of
our approach. Our algorithm is able to solve problems which are out of reach of
the state-of-the-art implementations. We give detailed information on its execution, in
particular the maximum number~$N$ observed at each level of the recursion.

\subsubsection*{Related works}
The study of properties of critical loci of generic quadratic or linear maps
that can be exploited for polynomial system solving is initiated in
\cite{BankGiustiHeintzMbakop_1997,BankGiustiHeintzPardo_2005} and followed by \cite{BankGiustiHeintzSafeyElDinSchost_2010}. Properness properties are
already exploited in algorithms of real algebraic geometry in
\cite{SafeyElDinSchost_2003}. The complexity analysis uses results
about the complexity of the geometric resolution algorithm \cite[and references therein]{GiustiLecerfSalvy_2001}.
The design of computer algebra algorithms whose practical behavior
depends more on the real geometry of the set under study is already exploited in
\cite{MS06} for answering connectivity queries (see also \cite{CapcoDinSchicho_2020} for
more recent developments).

\section{Geometry}\label{sec:geometry}

Let~$F = (f_1,\dotsc,f_s)$ be a polynomial map~$\mathbb{R}^n \to \mathbb{R}^s$.
In this whole section, we denote by $V\subset \R^n$ the real algebraic set
$F^{-1}(0)$.
Let~$h \in \mathbb{R}[x_1,\dotsc,x_n]$ be another polynomial.
For $t\in \R$, let $V(t) \eqdef V \cap h^{-1}(t)$.
We study the dimension of~$V$ in terms of the dimension of the fiber~$V(t)$.

\subsection{Dimension of fibers}

The dimension of~$V$ can be related to the dimension of the fibers~$V(t)$.
In the following statements, we use the convention~$\dim \varnothing = -1$.

\begin{proposition}\label{prop:dim-fiber-formula}
  Let~$Z \subset \mathbb{R}$ be a finite set such that~$t\mapsto
  \dim V(t)$ is locally constant on~$\mathbb{R}\setminus Z$.
  If\/~$V\neq \varnothing$, then
  \[ \dim V = \max \big( \max_{r\in Z} \dim V(r), \max_{t \in
  \mathbb{R}\setminus Z} \dim V(t) + 1 \big). \]
\end{proposition}

\begin{proof}

  By Hardt's triviality
  theorem \parencite{Hardt_1980}, there is a finite set~$Z'\subset \mathbb{R}$ such that~$h$ induces a semialgebraic locally trivial fibration
  on~$\mathbb{R}\setminus Z'$: if~$U$ is a connected component
  of~$\mathbb{R}\setminus Z'$, then~$V \cap h^{-1}(U)$ is isomorphic, as a semialgebraic set,
  to~$U\times V(t)$, for any~$t\in U$.
  Since~$V$ is the
  finite union of all~$V(r)$, for~$r \in Z'$, and all~$V\cap h^{-1}(U)$, for all
  connected components~$U$ of~$\mathbb{R}\setminus Z'$, we obtain that
  \begin{equation*}
    \dim V = \max \big( \max_{r\in Z'} \delta(r), \max_{U} \dim (
    V\cap h^{-1}(U)) \big).
  \end{equation*}
  Since $ V\cap h^{-1}(U) \simeq U\times V(t)$ for any~$t\in U$,
  we have $\dim ( V\cap h^{-1}(U) ) = \delta(t) +1$, unless it is empty.
  When it is empty, then~$\delta(t)+1 = 0$ holds, and since~$\dim V \geq
  0$, it holds
  that
  \begin{equation}\label{eq:9}
    \dim V = \max \big( \max_{r\in Z'} \delta(r), \max_{t\in
    \mathbb{R}\setminus Z'} \delta(t) + 1 \big),
  \end{equation}
  which is the claim, albeit with~$Z'$ in place of $Z$.

  Without loss of generality, we may assume that~$Z \subseteq Z'$.
  We now show that we can remove points from~$Z'$ without breaking Equation~\eqref{eq:9} as long as $\dim V(t)$ is locally constant on~$\mathbb{R} \setminus Z$.
  If~$Z = Z'$, there is naturally nothing to prove.
  Assume that~$Z' = Z \cup \left\{ a \right\}$ for some~$a\not\in Z$ (and the general case follows by induction).
  Let~$\delta(t)$ denote~$\dim V(t)$.
  Since~$\delta$ is locally constant on~$\mathbb{R} \setminus Z$, and~$a\not\in Z$, $\delta$ is constant in a neighborhood of~$a$.
  So
  \begin{equation*}\label{eq:4}
    \max_{t\in \mathbb{R} \setminus Z'} \delta(t) + 1= \max_{t\in \mathbb{R} \setminus Z} \delta(t) + 1.
  \end{equation*}
  Combining with \eqref{eq:9}, this gives
    \begin{align*}
    \dim V &= \max \big( \max_{r\in Z} \delta(r), \delta(a), \max_{t\in \mathbb{R} \setminus Z} \delta(t) + 1 \big).
  \end{align*}
  To conclude, we observe that
  \begin{equation*}\label{eq:8}
    \delta(a) \leq \delta(a)+1 \leq \max_{t \in \mathbb{R}\setminus Z} \delta(t) + 1
  \end{equation*}
  because $a \in \mathbb{R}\setminus Z$.
\end{proof}

\subsection{Generic case}

Under a genericity assumption on~$h$, the special fibers above~$Z$ in Proposition~\ref{prop:dim-fiber-formula} carry no information.

Let $\mathcal{L}$ be the set of linear forms on~$\mathbb{R}^n$
and let~$\mathcal{Q}$ be the set of quadratic functions~$x\mapsto \|x-p\|^2$
for~$p\in \mathbb{R}^n$. Both are real algebraic varieties isomorphic to~$\mathbb{R}^n$.

\begin{proposition}\label{prop:dimfiber}
  Let $S\subset \R^n$ be a semialgebraic set of positive dimension.
  If~$h$ is a generic element of~$\mathcal{L}$ or~$\mathcal{Q}$, then for any~$t \in \mathbb{R}$,
  \[\dim \big(S\cap h^{-1}(t)\big) < \dim(S) .\]
\end{proposition}

\begin{proof}
  Let~$d = \dim S$.
  Let~$W$ be the Zariski closure of~$S$ in $\mathbb{R}^n$.
  Note that~$\dim W = \dim S$  \parencite[Proposition~2.8.2]{BochnakCosteRoy_1998}.
  It is enough to prove that for any irreducible component~$Y$ of~$W$,
  \begin{equation}\label{eq:5}
    \dim \big( Y\cap h^{-1}(t) \big) < d.
  \end{equation}
  Indeed, by the inclusion~$S\subseteq W$,
  and the formula for the dimension of a union \parencite[Proposition~2.8.5(i)]{BochnakCosteRoy_1998}, we have
  \[ \dim \big(S\cap h^{-1}(t)\big) \leq \max_Y \dim \big( Y\cap h^{-1}(t) \big). \]

  Let~$Y$ be an irreducible component of~$W$. We may assume that $\dim Y \geq 1$
  as \eqref{eq:5} is trivial otherwise (because~$d \geq 1$).
  In particular, $Y$ contains at least two points~$p$ and~$q$.
  When~$h\in \mathcal{L}$ (resp.~$\mathcal{Q}$) is generic, it is clear that~$h(p)\neq h(q)$.
  In particular, $h$ is not constant on~$Y$ and therefore~$h-t$ is a nonzero function on~$Y$ for any~$t\in \mathbb{R}$.
  By irreducibility of~$Y$ and Krull's principal ideal theorem, it follows that~$ \dim \big( Y\cap h^{-1}(t) \big) < \dim Y$,
  which gives~\eqref{eq:5} and concludes the proof.
\end{proof}

\begin{proposition}\label{prop:dim-fiber-formula-generic}
  Assume that~$h$ is a generic element in~$\mathcal{L}$ or~$\mathcal{Q}$.
  Let~$Z \subset \mathbb{R}$ be a finite set such that~$t\mapsto
  \dim V(t)$ is locally constant on~$\mathbb{R}\setminus Z$.
  If\/~$V\neq \varnothing$, then
  \[ \dim V = \max_{t \in \mathbb{R}\setminus Z} \dim V(t) + 1. \]
\end{proposition}

\begin{proof}
  Proposition~\ref{prop:dimfiber} implies that
  $\dim V > \max_{r \in Z} \dim V(r)$.
  We conclude with Proposition~\ref{prop:dim-fiber-formula}.
\end{proof}

\subsection{Limits of critical values}
\label{sec:limits-crit-valu}

The main obstacle in applying Propositions~\ref{prop:dim-fiber-formula}
or~\ref{prop:dim-fiber-formula-generic} is the computation of a suitable finite
set~$Z \subset \mathbb{R}$ such that the dimension of the fiber~$V(t)$ is
locally constant outside~$Z$. The proof of
Proposition~\ref{prop:dim-fiber-formula} indicates that an effective Hardt
trivialization would be enough. Actually, much less is required. We show now how
to compute such a finite set $Z$ by means of critical points computations.

We start by introducing some notation. For $e \in \R^s$, let $S_e$ be the
semi-algebraic set
\begin{equation*}
  S_{e} \eqdef \left\{ p\in \mathbb{R}^n \st \abs{f_i(p)} \leq e_i \text{ for } 1\leq i \leq s \right\}.\label{eq:2}
\end{equation*}
For $t\in \R$, let~$S_e(t) \eqdef S_e \cap h^{-1}(t)$.

We will see that the dimension of~$V(t)$ is entirely determined by the homotopy
type of~$S_e(t)$, for a generic and small enough~$e\in \mathbb{R}^s$. Moreover,
the variations of the homotopy types of~$S_e(t)$ are controlled by Thom's
isotopy lemma when $t$ varies in terms of the critical values of~$h$ on a
Whitney stratification of~$S_e(t)$. Such a stratification is easy to get
when~$e$ is generic.

In all this section, we assume that for~$e$ in a neighborhood of~$0$, the
restriction of~$h$ to~$S_e$ is proper. In particular, this is the case when~$F$
is proper on~$\mathbb{R}^n$ (in which case~$S_e$ is compact), or when~$h$ is
proper on~$\mathbb{R}^n$.

Let~$e\in \mathbb{R}^s$ and let~$H_e$ be the
hypercube~$[-e_1,e_1]\times \dotsb\times [-e_s, e_s]$ (if~$e_i < 0$, then~$[-e_i, e_i]$ means~$[e_i, -e_i]$). The relative interior of
the $3^s$ facets of~$H_e$ form a Whitney stratification, which we
denote~$\whitney{{H}_e}$; namely
\[ \whitney{{H}_e} \eqdef \big\{ I_1\times \dotsb \times I_s \ \big|\  I_i \in \left\{ \left\{ -e_i \right\}, \left\rbrack -e_i, e_i \right\lbrack, \left\{ e_i \right\} \right\} \big\}. \]
When~$e$ is generic, the
pullback~$\whitney{{S}_e} \eqdef F^{-1}(\whitney{{H}_e})$ induces a Whitney stratification
of~$S_e$.
For this, it is enough to check that~$F$ is \emph{transverse} to each strata of~$\whitney{{H}_e}$ \parencite[(1.4)]{GibsonWirthmullerPlessisLooijenga_1976}.
The map~$F$ is transverse to a submanifold~$A \subseteq \mathbb{R}^s$ if
for any~$x \in F^{-1}(A)$,
\begin{equation}\label{eq:1}
  \ud_x F (\mathbb{R}^n) + T_{f(x)} A = \mathbb{R}^s.
\end{equation}
Let~$\mathcal{E} \subseteq \mathbb{R}^s$ be the set of all~$e$ such that
$\whitney{{S}_e}$ is a Whitney stratification of~$S_e$. The set $\mathcal{E}$ contains
a dense Zariski open set in~$\mathbb{R}^s$ \cite[General position
lemma]{Canny_1988}. In particular, if $e \in \mathcal{E}$, then~$se \in
\mathcal{E}$ for all but finitely many~$s \in \mathbb{R}$.

For~$e \in \mathbb{R}^s$, let~$\Sigma_{e}$ be
\[ \Sigma_e \eqdef \cup_{A\in \whitney{{S}_e}} h \big( \crit(h, A^\text{reg}) \big), \]
where~$A^\text{reg}$ denotes the regular locus of~$A$ (the points of~$A$ where the derivatives of the equations defining~$A$ are linearly independent).
When~$\whitney{S_e}$ is a Whitney stratification, then, by definition, the strata as regular and~$A = A^\text{reg}$,
so that
\[ \Sigma_e \eqdef \cup_{A\in \whitney{{S}_e}} h \big( \crit(h, A) \big). \]
By Sard's theorem, $\Sigma_e$ is finite.
We may consider the limit~$\lim_{r\to 0} \Sigma_{re}$. It is the finite set
\begin{equation}\label{eq:def-limit}
  \lim_{r\to 0} \Sigma_{re} = \cap_{R > 0} \overline{\cup_{0 < r < R} \Sigma_{re}},
\end{equation}
where the bar denotes closure with respect to Euclidean topology. This core
construction will be heavily used in the algorithms we describe in the next
section.


\begin{lemma}\label{lem:dimension-isotopy}
  Two real compact algebraic sets that are homotopically equivalent
  have the same dimension.
\end{lemma}

\begin{proof}
  Since the homology is invariant under homotopy equivalence \parencite[Corollary~2.11]{Hatcher_2002}, it is enough to
  give a homological characterization of dimension.

  Let~$W$ be a compact real algebraic set. We claim that the dimension
  of~$W$ is the
  largest integer~$d$ such that~$H_d(W; \mathbb{Z}_2)$ is not zero
  (where~$\mathbb{Z}_2 = \mathbb{Z}/2\mathbb{Z}$). Indeed, $W$ is homeomorphic
  to a $d$-dimensional simplicial complex
  \parencite[Theorem~9.2.1]{BochnakCosteRoy_1998}. In particular,
  $H_{k}(W; \mathbb{Z}_2) = 0$ for~$k > d$. Moreover, the sum of the
  $d$-dimensional simplices of a given simplicial decomposition of~$W$ is a
  nonzero element in~$H_d(W;\mathbb{Z}_2)$
  \parencite[Proposition~11.3.1]{BochnakCosteRoy_1998}, and then
  $H_d(W;\mathbb{Z}_2)\neq 0$.
\end{proof}

The next result establishes that we can take $Z = \lim_{r
  \to 0}\Sigma_{re}$.

\begin{theorem}\label{thm:dimension-locally-constant}
  Let~$e \in \mathcal{E}$ such that the restriction of~$h$ to~$S_e$ is
  proper. The dimension of\/~$V(t)$ is locally constant
  on~$\mathbb{R} \setminus \lim_{r \to0}\Sigma_{re}$.
\end{theorem}

\begin{proof}
  Let~$U \subseteq \mathbb{R}$ be a closed interval such that~$U \cap
  \lim_{r\to0}\Sigma_{re}$ is empty.
  There exists~$0< r < 1$ such
  that~$re \in \mathcal{E}$ and $\Sigma_{re}\cap U = \varnothing$.
  For notations, we may replace~$e$ by~$re$ and assume that~$r=1$.

  The sets~$A \cap h^{-1}(U)$, for~$A \in \whitney{{S}_{e}}$, form a
  Whitney stratification of~$S_e \cap h^{-1}(U)$.
  By construction, $\crit(h, A\cap h^{-1}(U))$ is empty.
  Indeed~$e\in \mathcal{E}$, so $\whitney{S_e}$ is a Whitney stratification and~$A = A^\text{reg}$. Moreover,
  \begin{align*}
  \crit\big(h, A^\text{reg} \cap h^{-1}(U)\big) &= \crit(h, A^\text{reg}) \cap h^{-1}(U) \\& \subseteq h^{-1}\big( \Sigma_{e} \cap U \big) = \varnothing.
  \end{align*}
  So the restriction of~$h$ to each~$A \cap h^{-1}(U)$
  is a submersion.
  Besides, the restriction of~$h$ to~$S_{e}$ is proper, by hypothesis.
  Therefore, by Thom's first isotopy lemma \parencite[Proposition~11.1]{Mather_2012},
  $h$ induces a locally trivial fibration over~$U$.
  Since~$U$ is an interval, it is simply connected and the fibration is trivial.
  If we choose a base point~$b\in U$, this means
  that there is a diffeomorphism~$\phi :  S_e\cap h^{-1}(U) \to U\times
  S_e(b)$
  such that~$\text{proj}_1 \circ \phi = h$.
  It follows easily that all~$S_e(t)$, for~$t \in U$ are diffeomorphic and,
  in particular, homotopically equivalent.

  Next, for any~$t\in U$, the inclusion~$V(t) \to S_e(t)$ is a
  homotopy equivalence \parencite[Proposition~1.6]{Durfee_1983}.
  So for any~$t,s\in U$, we have the homotopy equivalence
  \[ V(t) \mathop{\simeq}_{\text{Durfee}} S_e(t)  \mathop{\simeq}_{\text{Thom}}  S_e(s) \ \mathop{\simeq}_{\text{Durfee}} V(s). \]
  So the homotopy type of~$V(t)$ is constant for~$t\in U$.
  By Lemma~\ref{lem:dimension-isotopy}, the dimension of~$V(t)$ is constant
  for~$t\in U$.
\end{proof}

\section{Algorithms}\label{sec:algo}

An algorithm for computing the dimension mostly follows from
Proposition~\ref{prop:dim-fiber-formula-generic} and
Theorem~\ref{thm:dimension-locally-constant}. We first explain how to check a
sufficient condition for~$e$ to be in~$\mathcal{E}$.

\subsection{Genericity of the stratification}
\label{sec:check-gener-strat}

\begin{algo}[tp]
  \begin{description}
    \item[Input] $f_1,\dotsc,f_s \in \mathbb{R}[x_1,\dotsc,x_n]$, $e\in \mathbb{R}^s$
    \item[Output] \textsc{true} or \textsc{false}
    \item[Postcondition] If \textsc{true} is returned, then~$\whitney{{S}_e}$ is a Whitney stratification of~$S_e$.
          If~$e$ is generic, then returns \textsc{true}.
  \end{description}

  \begin{algorithmic}[1]
    \Function{CheckWhitneyStratification}{$f_1,\dotsc,f_s$, $e$}
    \For{$I \subseteq \left\{ 1,\dotsc,s \right\}$ with~$\#I \leq n+1$ and $\sigma : I \to \left\{ \pm 1 \right\}$}
    \State $M \gets \left( \partial f_i / \partial x_j \right)_{i \in I, j\in \left\{ 1,\dotsc,n \right\}}$, a $\#I\times n$ matrix
    \State $J \gets \langle f_i - \sigma_i e_i, i\in I\rangle + \langle \#I \times \#I \text{ minors of } M \rangle$
    \If{$J \neq \langle 1 \rangle$}
    \State \Return \textsc{false}
    \EndIf
    \EndFor
    \State \Return \textsc{true}
    \EndFunction
  \end{algorithmic}

  \caption{Checking the genericity of the stratification}
  \label{algo:check-whitney}
\end{algo}

We use the notations of \S\ref{sec:limits-crit-valu}. To check
that~$\whitney{{S}_e}$ is a Whitney stratification of~$S_e$, it is enough to check
Condition~\eqref{eq:1} for any strata~$A$ of~$\whitney{{H}_e}$. In order to reduce to
easier algebraic computations, we use instead the Zariski closure of the strata
and check the condition over~$\mathbb{C}$. This gives a stronger set of
conditions, that are are still sufficient conditions for~$\whitney{{S}_e}$ to be a
Whitney stratification.

A stratum~$A$ of~$\whitney{{H}_e}$ is determined by a subset~$I \subseteq \left\{ 1,\dotsc,s \right\}$
and a map~$\sigma : I \to \left\{ \pm 1 \right\}$.
$I$ and~$\sigma$ define a linear variety~$L_{I, \sigma} \eqdef \left\{ y\in\mathbb{C}^s \st y_i = \sigma_i e_i, i\in I \right\}$
which is the Zariski closure of the corresponding strata of~$\whitney{{H}_e}$, namely~$L_{I,\sigma} \cap H_e$.
The tangent space of~$L_{I,\sigma}$ at any~$y\in L_{I,\sigma}$ is $\left\{ y\in \mathbb{C}^s \st y_i = 0, i\in I \right\}$.
Therefore, if~$p_I$ is the linear projection~$\mathbb{R}^n\to \mathbb{R}^I$ retaining only coefficients whose indices are in~$I$,
Condition~\eqref{eq:1} is equivalent to~$p_I \circ \ud_x F : \mathbb{R}^n\to \mathbb{R}^I$ being surjective.
Therefore, if the variety
\begin{align*}
  Y_{I,\sigma}&\eqdef \left\{ x \in \mathbb{C}^n \st F(x) \in L_{I,\sigma} \text{ and } p_I \circ \ud_x F \text{ not surjective} \right\}.
\end{align*}
is empty, then Condition~\eqref{eq:1} holds for the strata defined by~$I$ and~$\sigma$.
The equations defining $Y_{I,\sigma}$ are~$f_i = \sigma_i e_i$, for~$i\in I$, and the vanishing of the $\#I\times\#I$ minors of~$p_I \circ \ud_x F$ (which is identified with a~$\#I \times n$ matrix with polynomial coefficients).

There are~$3^s$ possible values for~$I$ and~$\sigma$. However, when~$s > n+1$,
not all~$Y_{I,\sigma}$ need to be checked for emptyness. Indeed, when~$\#I > n$,
the map~$p_I \circ \ud_x F$ is never surjective (and the corresponding matrix
has no~$\#I\times\#I$ minors). In particular, $Y_{I,\sigma}$
contains~$Y_{I',\sigma'}$ for any~$I'$ and~$\sigma'$ such that~$I\subseteq I'$
and~$\sigma = \sigma'|_I$. Therefore, to check
that~$Y_{I, \sigma} = \varnothing$ for all~$I$ and~$\sigma$, it is enough to
check it for all~$I$ with~$\#I \leq n+1$. This leads to
Algorithm~\ref{algo:check-whitney} and the following statement.

\begin{proposition}
  On input~$f_1,\dotsc,f_s \in \mathbb{R}[x_1,\dotsc,x_n]$ and~$e \in \mathbb{R}^s$,
  Algorithm~\ref{algo:check-whitney} returns \textsc{true} if~$e$ is generic.
  If it returns \textsc{true}, then~$\whitney{{S}_e}$ is a Whitney stratification of~$S_e$.
\end{proposition}

\subsection{Computing the limits of critical values}

\begin{algo}[tp]
  \raggedright
  \begin{description}
    \item[Input] $f_1,\dotsc,f_s$ and~$h \in \mathbb{R}[x_1,\dotsc,x_n]$ and~$e \in (\mathbb{R} \setminus \left\{ 0 \right\})^s$
    \item[Precondition] The restriction of~$h$ to~$S_e = \left\{ |f_i| \leq \abs{e_i} \right\}$ is proper.
    \item[Output] A finite set~$Z \subset \mathbb{R}$ description as the zero set of a univariate polynomial
    \item[Postcondition] $\lim_{r\to 0} \cup_{A \in \whitney{{S}_{re}}} h \left( \crit(h, A^\text{reg}) \right) \subseteq Z$
  \end{description}

  \begin{algorithmic}[1]
    \Function{LimitCriticalValues}{$f_1,\dotsc,f_s$, $h$, $e$}
    \State $Z \gets \varnothing$
    \For{$I \subseteq \left\{ 1,\dotsc,s \right\}$ with~$\#I \leq n$ and $\sigma : I \to \left\{ \pm 1 \right\}$}
    \State $J' \gets \langle \sum_{i\in I} \lambda_i \partial_j f_i - \partial_j h \rangle_{j\in [n]} + \langle \sigma_j e_j f_i - \sigma_i e_i f_j\rangle_{i,j\in I}$
    \State\Comment The~$\lambda_i$ are new variables. This is the ideal of~$W'_{I, \sigma, e}$.
    \State $J \gets J' \cap \mathbb{R}[x_1,\dotsc,x_n] + \langle f_1,\dotsc,f_s \rangle$
    \State \Comment Ideal of~$V\cap \overline{W_{I,\sigma,e}}$
    \State $ p\gets$ a generator of~$\left( J + \langle h-t\rangle \right) \cap \mathbb{R}[t]$
    \State $Z \gets Z \cup p^{-1}(0)$
    \EndFor
    \State \Return $Z$
    \EndFunction
  \end{algorithmic}

  \caption{Limits of critical values}
  \label{algo:limits-critical-values}
\end{algo}

\subsubsection{Description}

We consider the problem of computing the limit set~$\lim_{r\to 0} \Sigma_{re}$, for a given~$e \in (\mathbb{R} \setminus \{0\})^s$,
using the notations of~\S \ref{sec:limits-crit-valu}.
To this purpose, we extend to the case of several equations the method of \parencite{SafeyElDin_2005}, and remove the genericity assumptions. Most of the technical difficulties come from considering several equations.

Let~$P_e$ denote the critical points of~$h$ on the strata of~$\whitney{{S}_e}$, that is
$P_e \eqdef \cup_{A\in \whitney{{S}_e}} \crit(h, A^\text{reg})$,
so that~$\Sigma_e = h(P_e)$.
The limit set~$\lim_{r\to 0} P_{re}$ is defined as in~\eqref{eq:def-limit}.

\begin{lemma}\label{lem:limit-critical-values-points-exchange}
  Let~$e\in \mathbb{R}^s$ such that~$h$ is proper on~$S_e$.
  Then
  \[ \lim_{r\to 0} \Sigma_{re} = h \big( \lim_{r\to 0} P_{re} \big). \]
\end{lemma}

\begin{proof}
  This commutation of~$h$ and~$\lim_{r\to 0}$ is an elementary property of
  proper maps. The right-to-left inclusion follows directly from the
  definitions. Conversely, let~$y\in \lim_{r\to 0} \Sigma_{re}$.
  Let~$(R_i)_{i \geq 1}$ be a sequence decreasing to~0, with~$R_1 \leq 1$. For
  any~$i \geq 1$, $y \in \overline{\cup_{0<r<R_i} h(P_{re})}$, by definition. So there is
  some~$x_i \in \cup_{0< r < R_i} P_{re}$ such
  that~$\abs{y-h(x_i)} \leq 1/i$.
  Note that~$P_{re} \subseteq S_{re}$, so~$x_i \in S_e$ (using~$R_i \leq 1$).
  In particular,
  $ x_i \in h^{-1}([y-1,y+1]) \cap S_e$, which is a compact set, by hypothesis.
  Up to extracting a subsequence, we may assume that~$(x_i)_i$ converges to some~$x\in \mathbb{R}^n$.
  By continuity, $h(x) = y$.
  Since~$R_j \leq R_i$ for~$j \geq i$, it follows that~$x \in \overline{\cup_{0 < r < R_i} P_{re}}$
  for any~$i\geq 1$. This means that~$x \in \lim_{r\to 0} P_{re}$.
\end{proof}

Let~$Q_e$ be the complex analogue of~$P_e$, that is the union over all~$I$ and~$\sigma$ of the complex critical points of~$h$ restricted to the regular locus of the complex stratum~$F^{-1}(L_{I, \sigma})$. In other words,
\begin{equation*}
Q_e \eqdef \cup_{I, \sigma} \crit\big(h, \left\{ f_i=\sigma_i e_i \text{ for } i\in I \right\}^\text{reg} \big) \subseteq \mathbb{C}^n.
\end{equation*}
It is clear that~$P_e \subseteq Q_e$.
Moreover, let~$W_e \eqdef \cup_{r\in \mathbb{C}^*} Q_{re}$.
If~$x \in Q_{re}$, then~$r$ is entirely determined by~$x$ and~$e$.
Namely, if~$x$ lies in the stratum~$F^{-1}(L_{I,\sigma})$, then~$r = \sigma_i f_i(x)/e_i$, for any~$i\in I$.
This gives a well-defined regular map~$\rho : W_e \to \mathbb{C}$ such that~$x\in Q_{\rho(x) e}$ for any~$x\in W_e$.

\begin{proposition}\label{prop:limit-critical-points}
  For any~$e\in (\mathbb{R}\setminus \{0\})^s$,
  the set
   $h( V \cap  \overline{ W_{e} } )$
   is finite.
   Moreover, if~$h$ is proper on~$S_e$, then~$\lim_{r\to 0} \Sigma_{re} \subseteq h(V \cap \overline{W_e})$.
\end{proposition}

\begin{proof}
  We first prove the inclusion, when~$h$ is proper on~$S_e$.
  Let~$y\in \lim_{r\to 0} \Sigma_{re}$. By Lemma~\ref{lem:limit-critical-values-points-exchange},
  there is some~$x\in \lim_{r\to 0} P_{re}$ such that~$y = h(x)$.
  By definition, $x \in \overline{\cup_{0<r<R} P_{re}}$ for any  $R > 0$.
  Since~$P_{re} \subseteq S_{re}$, it follows that~$\abs{f_i(x)} \leq R \abs{e_i}$ for any~$R> 0$ and any~$i$.
  Therefore, $f_i(x) = 0$ for any~$i$, and~$x\in V$.
  Moreover, since~$P_{re} \subseteq Q_{re} \subseteq W_e$, we have~$x \in \overline{W_e}$.

  For the finiteness,
  consider the map~$\psi : x\in W_e \to \mathbb{C}^2$
  defined by~$\psi(x) = ( \rho(x), h(x))$.
  Let~$C = \psi(W_e)$.
  For any~$r \in \mathbb{C}$, the set~$C\cap \left\{ z_1 = r \right\}$ is finite (where $z_1$ and~$z_2$ are the coordinates on~$\mathbb{C}^2$). Indeed,
  \begin{equation*}
    C\cap \left\{ z_1 = r \right\} = \psi \left( W_e \cap \rho^{-1}(r) \right) = \left\{ r \right\} \times h \left( Q_{re} \right),
  \end{equation*}
  and Sard's theorem implies that~$h(Q_{re})$ is finite.
  It follows that the set~$\overline{C} \cap \left\{ z_1 = r \right\}$ are finite too.
  (Otherwise $\overline{C}$ would contain a vertical line~$\left\{ z_1=r \right\}$ and so~$C$, which is dense in all the components of~$\overline{C}$, would contain a dense subset of this line, which would contradict the finiteness of~$C\cap \left\{ z_1 = r \right\}$.)

  Lastly, we observe that
  \begin{align*}
    h \big( V \cap \overline{W_{e}} \big) &= h \big( \overline{W_e} \cap \rho^{-1}(0) \big) = \mop{proj}_2 \big( \psi \big( \overline{W_{e}} \big) \cap \big\{ z_1 = 0 \big\} \big) \\
    &\subseteq \mop{proj}_2 \big( \overline{C} \cap \big\{ z_1 = 0 \big\} \big),
  \end{align*}
  which is finite.
\end{proof}

Based on this statement, we give two ways of computing a finite set
containing~$\lim_{r\to 0} \Sigma_{re}$. One with Gröbner basis, without
genericity hypothesis, and another using geometric resolution under genericity
hypothesis, in order to obtain a complexity estimate.

\subsubsection{Using Gröbner bases}\label{ssec:grobner}

For~$I$ and~$\sigma$ defining a stratum of~$\whitney{{S}_e}$, 
consider the algebraic subvariety of~$\mathbb{C}^n\times \mathbb{C}^{\# I}$
\begin{multline*}\textstyle
  W'_{I, \sigma, e} \eqdef \big\{ (x, \lambda)\mid \ud_x h = \sum_{i\in I} \lambda_i \ud_xf_i \\ \text{ and } \forall i,j \in I, \sigma_i e_i f_j = \sigma_j e_jf_i \big\}.
\end{multline*}
Let also~$W'_e \eqdef \cup_{I, \sigma} W'_{I, \sigma, e}$.

\begin{lemma}\label{lem:complex-critical-points}
  For any~$e\in (\mathbb{C} \setminus \left\{ 0 \right\})^s$, $W_{e} = \mop{proj}_1 (W'_e)$.
\end{lemma}

\begin{proof}
  The left-to-right inclusion follows directly from the definitions.
  Conversely, let~$(x,\lambda) \in \cap W'_{I, \sigma, e}$.
  The equations~$\sigma_i e_i f_j = \sigma_j e_jf_i$ imply the existence of a unique~$r \in \mathbb{C}$
  such that~$f_i(x) = \sigma_i  r e_i$ for all~$i\in I$.
  Furthermore, we can choose~$\lambda$ such that the number of nonzero~$\lambda_i$ is minimized.
  Let~$J = \left\{ i\in I\st \lambda_i \neq 0 \right\}$.
  By minimality of~$J$, the derivatives~$\ud_x f_i$, for~$i\in J$, are linearly independent.
  In particular~$x$ is a regular point of the complex stratum~$\left\{ f_i=\sigma_i re_i \text{ for } i\in I \right\}$ of~$\whitney{{S}_{re}}$
  and it is a critical point of~$h$ restricted to the regular locus of this variety.
  So~$x \in Q_{re} \subset W_e$.
\end{proof}

This leads to  Algorithm~\ref{algo:limits-critical-values} for computing a finite set containing the limit set $\lim_{r\to 0} \Sigma_{re}$.
As an important optimization when the number of equations is large, note that if~$\#I \geq n$, then~$W'_{I, \sigma, e} = W'_{J, \sigma|_J, e}$ for some~$J \subseteq I$ with~$\#J = n$. Indeed, if~$\ud_x h = \sum_{i\in I} \lambda_i \ud_x f_i$, then we can find a similar relation using at most~$n$ derivatives~$\ud_x f_i$.

\begin{proposition}\label{prop:correctness-algorithm-crit}
  On input~$f_1,\dotsc,f_s$, $h$ and~$e \in (\mathbb{R}\setminus \{0\})^s$, and assuming that~$h$ is proper on~$S_e$,
  then Algorithm~\ref{algo:limits-critical-values} returns a finite set~$Z$ containing~$\lim_{r\to 0} \Sigma_{re}$.
\end{proposition}

\subsubsection{Complexity}
\label{sec:complexity-limit}

When~$h$ is generic, the computation of limits of critical values can be performed in the framework of geometric resolution, which leads to complexity bounds.
For brevity, we study the case $s=1$; without much loss of generality since we can replace several equations with a sum of squares.

\begin{proposition}[{\cite{SafeyElDin_2005}}]
  \label{prop:limit-complexity}
  On input~$f \in \mathbb{R}[x_1,\dotsc,x_n]$, $h \in \mathcal{L}$ or $\mathcal{Q}$ generic and~$e\in \mathbb{R}^s$ generic,
  one can compute a finite set~$Z \subset \mathbb{R}$ with less than~$D^{n}$ elements such that~$\lim_{r\to 0} \Sigma_{re} \subseteq Z$
  in at most $\mop{poly}(\log D, n) D^{2n+2}$ arithmetic operations,
  where~$D = \deg f$ and $L$ is the evaluation complexity of~$f$.
\end{proposition}

\subsection{Computation of the dimension}

\begin{algo}[tp]
   \begin{description}
    \item[Input] $f_1,\dotsc,f_s \in \mathbb{R}[x_1,\dotsc,x_n]$
    \item[Precondition] The map~$x \mapsto (f_1(x),\dotsc,f_s(x))$ is proper.
    \item[Output] Real dimension of~$\left\{ x\in \mathbb{R}^n \st f_1(x) = \dotsb = f_s(x) = 0 \right\}$
  \end{description}

  \begin{algorithmic}[1]
    \Function{$\mathrm{Dim}_\mathrm{proper}$}{$f_1,\dotsc,f_s$}
    \If{$\left\{ f_1=\dotsb=f_s = 0 \right\} = \varnothing$} \label{line:test-emptiness}
    \State\Return $-1$ \label{line:empty-case}
    \EndIf

    \If{$n=1$}
    \State\Return 0
    \EndIf
    \State $e\gets $ a generic element of~$\mathbb{R}^s$ \label{line:pick-e}
    \State $q \gets $ a generic linear form in~$x_1,\dotsc,x_{n-1}$ \label{line:pick-h}
    \State $Z \gets \text{LimitCriticalValues}(f_1,\dotsc,f_s, x_n-q, e)$

    \State $\dim \gets -1$
    \For{$U$ connected component of~$\mathbb{R} \setminus Z$} \label{line:loop-cc}
    \State $t \gets $ some point in~$U$
    \State $\text{dimfiber} \gets \smash{\mathrm{Dim}_\mathrm{proper}(f_1|_{x_n\gets q + t},\dotsc,f_s|_{x_n \gets q + t})}$
    \State $\dim \gets \max( \dim, \text{dimfiber}+1)$
    \EndFor
    \State\Return $\dim$
    \EndFunction
  \end{algorithmic}

  \caption[]{{Dimension of a real algebraic set (proper case)}}
  \label{algo:dim-proper}
\end{algo}

\begin{algo}[tp]
  \begin{description}
    \item[Input] $f_1,\dotsc,f_s \in \mathbb{R}[x_1,\dotsc,x_n]$
    \item[Output] Real dimension of~$\left\{ p\in \mathbb{R}^n \st f_1(p) = \dotsb = f_s(p) = 0 \right\}$
  \end{description}

  \begin{algorithmic}[1]
    \Function{Dimension}{$f_1,\dotsc,f_s$}
    \If{$\left\{ f_1=\dotsb=f_s = 0 \right\} = \varnothing$} 
    \State\Return $-1$
    \EndIf

       \State $e\gets $ a generic element of~$\mathbb{R}^s$
       \State $p \gets $ a generic element of~$\mathbb{R}^n$
       \State $h \gets (x_1-p_1)^2+\dotsb+(x_n-p_n)^2$
       \State $Z \gets \text{LimitCriticalValues}(f_1,\dotsc,f_s, h, e)$

       \State $\dim \gets -1$
       \For{$U$ connected component of~$\mathbb{R} \setminus Z$}
       \State $t \gets $ some point in~$U$
    \State $\dim \gets \max( \dim, \smash{\mathrm{Dim}_\mathrm{proper}(f_1,\dotsc,f_s, h-t)}+1)$
    \EndFor
    \State\Return $\dim$

    \EndFunction
  \end{algorithmic}

  \caption[]{{Dimension of a real algebraic set}}
  \label{algo:dim}
\end{algo}

\begin{theorem}\label{thm:main:complexity}
  On input~$f_1,\dotsc,f_s$, and assuming that the map~$x\in \mathbb{R}^n\mapsto (f_1(x),\dotsc,f_s(x))$ is proper,
  Algorithm~\ref{algo:dim-proper} generically\footnote{That is, assuming that the points picked at lines~\ref{line:pick-e} and~\ref{line:pick-h} are generic enough.} returns
  the dimension of the real algebraic set~$\left\{ f_1=\dotsb=f_s=0 \right\}$
\end{theorem}

\begin{proof}
  We proceed by induction on~$n$.
  The case~$n=1$ is trivial.
  If~$V = \left\{ f_1=\dotsb=f_s= 0 \right\}$ is empty, then the algorithm returns $-1$ on line~\ref{line:empty-case}.
  Assume now that~$V$ is not empty.
  Using Proposition~\ref{prop:correctness-algorithm-crit} and Theorem~\ref{thm:dimension-locally-constant},
  the function $\dim V(t)$ is locally constant on~$\mathbb{R}\setminus Z$.
  So the algorithm computes and return~$\max_{t\in \mathbb{R}\setminus Z} \dim V(t)$.
  By Proposition~\ref{prop:dim-fiber-formula-generic}, this is~$\dim V$.
\end{proof}

The nonproper case is similar.

\begin{theorem}
    On input~$f_1,\dotsc,f_s$,
  Algorithm~\ref{algo:dim} generically returns the dimension of the real algebraic set~$\left\{ f_1=\dotsb=f_s=0 \right\}$.
\end{theorem}

We study the complexity in the proper case in the case~$s=1$ in the framework of geometric resolution.

\begin{theorem}
  Given~$f \in \mathbb{R}[x_1,\dotsc,x_n]$ (of degree~$D$ and evaluation
  complexity~$L$) such that~$\mathbb{R}^n \cap \left\{ f=0 \right\}$ is bounded.
  One can compute the dimension of~$\mathbb{R}^n \cap \left\{ f=0 \right\}$ in
  $\mop{poly}(\log D, n) L D^{n(d+3) + 1}$ arithmetic operation and at
  most~$D^{n(d+1)}$ isolation of the real roots of polynomials of degree at
  most~$D^n$, where~$d$ the dimension to be computed.
\end{theorem}

\begin{proof}
  First note that the emptiness test
  (line~\ref{line:test-emptiness}) can be performed
  with~$\mop{poly}(\log D, n) D^{2(n+1)}$ arithmetic operations \cite{SafeyElDin_2005}, comparable to the cost of computing the limits of critical values. Let~$a_{n,d}$ be the
  maximum number of arithmetic operations performed by the algorithm given as
  input~$s$ equations of degree~$D$ defining an algebraic set of dimension~$d$, excluding the root isolation necessary to idenitfy points in the connected components of~$\mathbb{R} \setminus Z$ (line~\ref{line:loop-cc}).
  We have $a_{n, -1} = \mop{poly}(\log D, n) D^{2(n+1)}$, and
  given that~$\# Z < D^n$, we have
  $a_{n, d} \leq \mop{poly}(\log D, n) D^{2(n+1)} + D^n a_{n-1, d-1}$.
  It follows that
  \begin{align*}
    a_{n, d} &=  L \sum_{k=0}^{d+1} D^{n+\dotsb+(n-k+1)} \mop{poly}(\log D, n) D^{2(n-k)+2}\\
             &\leq \mop{poly}(\log D, n) L D^{n(d+3) + 1}. \qedhere
  \end{align*}
\end{proof}

\subsection{Computation of the dimension (Las Vegas)}

\begin{algo}[tp]
   \begin{description}
    \item[Input] $f_1,\dotsc,f_s \in \mathbb{R}[x_1,\dotsc,x_n]$
    \item[Precondition] The map~$x \mapsto (f_1(x),\dotsc,f_s(x))$ is proper.
    \item[Output] Real dimension of~$\left\{ x\in \mathbb{R}^n \st f_1(x) = \dotsb = f_s(x) = 0 \right\}$
  \end{description}

  \begin{algorithmic}[1]
    \Function{$\mathrm{Dim}_\mathrm{LV}$}{$f_1,\dotsc,f_s$}
    \If{$\left\{ f_1=\dotsb=f_s = 0 \right\} = \varnothing$}
    \State\Return $-1$
    \EndIf
    \If{$n=1$}
    \State\Return $0$
    \EndIf

    \Repeat
    \State $e\gets $ a generic element of~$\mathbb{R}^s$
    \Until{$\text{CheckWhitneyStratification}(f_1,\dotsc,f_s,e)$}

    \State $\dim \gets -1$
    \For{$i\in \left\{ 1,\dotsc,n \right\}$}
    \State $Z \gets \text{LimitCriticalValues}(f_1,\dotsc,f_s, x_i, e)$

    \For{$U$ connected component of~$\mathbb{R} \setminus Z$}
    \State $t \gets $ some point in~$U$
    \State $\dim \gets \max( \dim, \smash{\mathrm{Dim}_\mathrm{LV}(f_1|_{x_i\gets t},\dotsc,f_s|_{x_i \gets t})}+1)$
    \EndFor
    \EndFor
    \State\Return $\dim$

    \EndFunction
  \end{algorithmic}

  \caption[]{{Dimension of a real algebraic set (Las Vegas)}}
  \label{algo:dim-lasvegas}
\end{algo}

In Algorithm~\ref{algo:dim-proper}, the genericity of~$e$ may be checked with Algorithm~\ref{algo:check-whitney}.
However, we do not know how to check the genericity of~$h$.
The problem can be circumvented by considering~$n$ linearly independent linear
forms~$h_1, \ldots, h_n$,
based on the following statement.

\begin{proposition}
  Let~$V_i(t) = V \cap \left\{ x_i = t \right\}$.
  Let~$Z_1,\dotsc,Z_n$ be finite sets such that $t\mapsto \dim V_i(t)$ is locally constant on~$\mathbb{R} \setminus Z_i$.
  If~$V \neq \varnothing$ then
  \begin{equation*}
    \dim V = \max_{1\leq i\leq n} \max_{t\in \mathbb{R} \setminus Z_i} \dim V_i(t) + 1.
  \end{equation*}
\end{proposition}

\begin{proof}
  By Proposition~\ref{prop:dimfiber}, $\dim V \geq \max_{t\in \mathbb{R} \setminus Z_i} \dim V_i(t) + 1$ for any~$i$.
  Assume, for contradiction, that $\dim V \geq \max_{t\in \mathbb{R} \setminus Z_i} \dim V_i(t) + 2$ for all~$i$.
  Let~$W$ the set of points where~$V$ has dimension~$\dim V$. This is a closed semialgebraic set of dimension~$d$ \cite[Proposition~2.8.12]{BochnakCosteRoy_1998}.
  Decomposing~$V$ with Hardt's triviality theorem, as in the proof of Proposition~\ref{prop:dimfiber},
  reveals that~$W \subseteq \left\{ x\in \mathbb{R}^n \st x_i \in Z_i \right\}$, for any~$i$.
  This implies that~$W \subseteq Z_1 \times \dotsb \times Z_n$, so~$W$ is finite and~$\dim V = \dim W = 0$.
  This contradicts~$\dim V \geq 2$.
\end{proof}

\section{Experiments}\label{sec:experiments}
\def\ffour{$F_4$\xspace}
\def\sparsefglm{Sparse-FGLM\xspace}
\def\fgb{\texttt{FGb}\xspace}
\def\msolve{\texttt{msolve}\xspace}
\def\maple{Maple\xspace}

\subsubsection*{Implementation}
We have implemented the Monte-Carlo version of our algorithm
(Algorithm~\ref{algo:dim}). We rely on implementations based on Gr\"obner bases
for saturating polynomial ideals by others and solving zero-dimensional
polynomial systems exactly by computing rational parametrizations of their
solution sets. We use implementations of algorithms such as \ffour
\cite{Faugere_1999} and \sparsefglm \cite{FaugereMou_2017} for computing
Gr\"obner bases using graded reverse lexicographical orderings (or elimination
algorithms) and change the ordering to a lexicographical one in the
zero-dimensional case. We rely on the libraries
\href{https://www-polsys.lip6.fr/~jcf/FGb/index.html}{\fgb}~\cite{Faugere_2010}
for computing Gr\"obner bases w.r.t. elimination orderings and
\href{https://msolve.lip6.fr}{\msolve}~\cite{msolve} for solving multivariate
polynomial systems. The algorithm itself is implemented using the Maple (version
2020) computer algebra system. In Table~\ref{table:timings}, we report on
timings obtained on several instance described below. All computations have been
performed sequentially using an Intel Xeon E7-4820 (2.00 GHz) with 1.5 Tb of RAM.

There are a few significant
variations between the text and the implementation, concerning
Algorithm~\ref{algo:limits-critical-values} for the most part.
Firstly, instead of
computing the limits of critical values, we compute the limits of critical
points. From them, we obtain not only the limits of critical values, but we also
decide the emptiness or nonemptiness of the input.
Secondly, and this only matters when~$s > 1$, instead
of introducing and eliminating the variables~$\lambda_i$ we directly deal with a
formulation in terms of the minors of the Jacobian matrix.
Thirdly, and this also matters only when~$s > 1$, we ignore $\sigma$.
Indeed, we observe that whenever~$Z$ and~$Z'$ are two subsets of~$\mathbb{R}$
satisfying the hypothesis of Proposition~\ref{prop:dim-fiber-formula-generic},
then so does~$Z\cap Z'$. By considering the intersection of all~$\lim_{r\to 0} \Sigma_{re}$
as the signs of the coefficients~$e_i$ run though all possible configurations, it easy to check that we can remove the loop over~$\sigma$ in Algorithm~\ref{algo:limits-critical-values} while preserving the correctness of Algorithm~\ref{algo:dim-proper}.

{\small
  \begin{table*}[t]
  \centering
    \begin{tabular*}{\textwidth}{lrrr @{\extracolsep{\fill}} rrrrrrrrrrr}
      \toprule
        & $n$ & $D$ & dim. &  \#fibers  in depth ...0 & ...1 & ...2 & ...3 &
                                                                             ...4
      & max. deg. & ours & CAD & RT & BPR & BS \\

        $p_4$   & 4 & 4& 3& 2 & 4 & 1&&&  16 & 0.4& {0.54}& \textbf{0.1}& -& 2.1\\
        $p_5$   & 5 & 4&2 & 10 & 10 & 1&&& 84 & 38& {25}& \textbf{3.2} &- & 42\\
        $p_6$   & 6 & 4& 2& 10 & 10 & 1&&& 292 & 786&- & \textbf{18} & -& 452\\
        $p_7$   & 7 & 4 & 2 & 10 & 10 & 1 &&& 940 & \textbf{9\,450} & - & 200\
                                                                          071 & - & - \\
        $p_8$   & 8 & 4 & 2 & 10 & 10 & 1 &&& 2568 & \textbf{538\,028} & - & - & - & - \\
        \midrule
        $b_4$   & 4 & 8& 1& 2 & 1 &&&& 400 & 43& -& \textbf{0.9} & -& {10}\\
        $b_5$   & 5 & 10 & 1 & 2 & 1 &&&& 2100 & 8\,280 & - & \textbf{16} & - & {650}\\
 \midrule
       $d_{3,5}$ & 5 & 6& 3& 8 & 14 & 12 & 1&& 264 & \textbf{416} & - & - & - & 739\\
       $d_{3,6}$ & 6 & 6 & 4 & 4 & 8 & 12 & 8 & 1& 288& \textbf{836} & - & - & - & 28\,867 \\
       $d_{3,7}$ & 7 & 6& 4 & 8 & 8 & 12 & 12 & 1& 288& \textbf{\,400} & - & - & - & - \\
       $d_{4,3}$  & 3 & 12& 1& 12 & 14&&&& 756 & \textbf{334}& - & - & - & 1\,816\\
       $d_{4,4}$  & 4 & 12& 2& 16 & 22 & 14&&& 2328& \textbf{31\,060}& -& - & -& -\\
 \midrule
       $s_{3,5}$  & 5 & 4& 2& 4 & 8 & 6&&& 176& \textbf{84}&- & - &- & 8\,000\\
       $s_{3,6}$  & 6 & 4& 3&8 & 16 & 10 & 1 && 296& \textbf{2\,144}& -& - & -& -\\
       $s_{3,7}$  & 7 & 4&4 & 6 & 12 & 12 & 10 & 1& 448& \textbf{6\,411}&- & - &- &- \\
       $s_{4,5}$  & 5 & 4& 1& 4 & 6&&&& 192& \textbf{108}&- & -&- & 4\,473\\
       $s_{4,6}$  & 6 & 4&2 & 4 & 8 & 6&&& 344& \textbf{1\,860}& -& -& -& -\\
       $s_{4,7}$  & 7 & 4&3 &10 & 14 & 14 & 8 && 344 & \textbf{35\,814}&-&- &- & -\\
       $s_{5,6}$  & 6 & 4 & 1& 2 & 2&&&& 688& \textbf{2\,329}& -& -& - & - \\
 \midrule
        $\text{Vor}_1$  &6 & 8 & 4 & 5 & 8 & 6 & 4 & 1 & 544 &\textbf{61\,705} &
                                                                                 - & -& -
                 & - \\
        $\text{Sottile}$  & 4 & 24& 2 & 1 & 12 & 18 &&& 4\,052& \textbf{11\,338}
                  & - & -& -
                 & - \\
      \bottomrule
  \end{tabular*}
  \medskip
\captionsetup{singlelinecheck=off}
  \caption[]{{Timings for computing the real dimension of several instances.\\
      Description of the columns: ``$n$'', the number of variables;
         ``$D$'', the degree of th input;
         ``\# fibers in depth $k$'', the maximum cardinality of $Z$ at depth~$k$ of the recursion;
         ``max. deg.'', the maximum degree of all zero-dimensional polynomial systems which are solved during the execution of the algorithm;
         ``ours'', timings of our algorithm, in seconds;
         ``CAD'', timings of Maple's CAD;
         ``RT'', timings of Maple's real triangularization,
         ``BPR'', timings of our implementation of~\cite{BasuPollackRoy_2006};
         ``BS'', timings of~\cite{BannwarthSafeyElDin_2015}. Symbol '-' means
         that the computation was stopped after $2$ weeks or $100$ times the best
         runtime achieved by another method.
  }}
\label{table:timings}
\end{table*}
}

\subsubsection*{Benchmarks}
We use the following instances to evaluate the efficiency of our algorithm
and its implementation, always in the case~$s=1$.


\paragraph{Family~$p_n$} These are the following polynomials of degree $4$:
\[
  p_n = \bigg( \sum_{i=1}^n x_i^2 \bigg)^2 - 4 \sum_{i=1}^{n-1}x_{i+1}^2 x_i^2
  - 4 x_1^2 x_n^2.\] These polynomials are sums-of-squares and then non-negative
over the reals.

\paragraph{Family~$b_n$} These are the following polynomials of degree $2n$ (where $n$
is the number of variables) \[
b_n = \prod_{i=1}^n\left( x_i^2+n-1 \right) - n^{n-2}\bigg( \sum_{i=1}^n
x_i \bigg)^2 .\] These polynomials were introduced in \cite{HanDaiXia_2014} and are
known to be non-negative over the reals as well.

\paragraph{Family~$s_{c,n}$} These are polynomials which are sums of squares. We denote
by $s_{c,n}$ a sum of squares of $c$ quadrics in $\Q[x_1, \ldots, x_n]$. All
these polynomials have degree $4$.

\paragraph{Family~$d_{k,n}$} These polynomials are discriminants of characteristic
polynomials of $k\times k$ symmetric linear matrices with entries in $\Q[x_1,
\ldots, x_n]$. Such polynomials are known to be sums-of-squares \cite{Lax_1998}.
Hence whenever it is non-empty, their real solution set has dimension less than
$n-1$. Further, these polynomials are denoted by $d_{k, n}$ (we take randomly
chosen dense linear entries in the matrix); they have total degree $2k$.

\paragraph{Other polynomials}
The example ``${\rm Vor}_1$'' comes from \cite{EverettLazardLazardSafeyElDin_2009}. The example ``Sottile'',
communicated to us by F.~Sottile, arises in enumerative geometry.

\subsubsection*{Results}

In Table~\ref{table:timings}, we report on results obtained by comparing the
implementation of our algorithm with implementations of the Cylindrical
Algebraic Decomposition in Maple 2020 (command
\texttt{CylindricalAlgebraicDecompose}), the algorithm of \cite{CDMMXX} which
decomposes semi-algebraic systems into triangular systems from which the
dimension can be read (command \texttt{RealTriangularize}), as well as our
implementations (using Maple 2020 again) of the algorithm of in
\cite{BasuPollackRoy_2006} and , which are both based on quantifier elimination
through the critical point method.

From a practical point of view, the worst algorithm is the one of
\cite{BasuPollackRoy_2006} which cannot solve any of the problems in our
benchmark suite, despite the fact that it provides the best theoretical
complexity $D^{O(d(n-d))}$. The main reason for that is the too large constant
hidden by the ``big-O'' notation which is here in the exponent.

On benchmarks $p_n$ and $b_n$, our method
suffers from the choice of generic quadratic forms or linear forms
compute limits of critical points. On these
examples, decomposition methods, and especially the real triangular decomposition \cite{CDMMXX},
perform well because the cells are simple.
For $p_7$ and $p_8$, the theoretical complexity takes over and our
algorithm is faster.

On all other instances of our benchmark suite, our algorithm is easer faster
than the state-of-the-art software or it can solve problems which were
previously out of reach. This is explained by the
fact that, as we expected, the number of fibers to be considered for the
recursive calls is lower by several orders of magnitude than the exponential
bound $D^n$. Hence, in practice, one observes a behavior which is far from the
complexity $D^{O(nd)}$ estimate.

\begin{acks}
  This work is supported by \grantsponsor{ANR}{ANR} under grants
  \grantnum{ANR}{ANR-19-CE48-0015} (ECARP), \grantnum{ANR}{ANR-18-CE33-0011}
  (Sesame) and \grantnum{ANR}{ANR-19-CE40-0018} (De rerum natura), and the
  \grantsponsor{H2020}{EU H2020 research and innovation programme} under grant
  \grantnum{H2020}{813211} (POEMA).
\end{acks}

\raggedright

\small
\printbibliography

\end{document}